\documentclass{scrartcl}

\KOMAoptions{paper=letter}

\RequirePackage{fix-cm}
\RequirePackage[utf8]{inputenc}
\RequirePackage[T1]{fontenc}
\RequirePackage{xspace}
\RequirePackage{amssymb}
\RequirePackage{amsmath}
\RequirePackage{amsthm}
\RequirePackage[hidelinks]{hyperref}
\RequirePackage[round]{natbib}

\RequirePackage{nicefrac}
\RequirePackage{tikz}
\usetikzlibrary{positioning,arrows,shapes,decorations,calc,fit,arrows.meta} 
\RequirePackage{booktabs}
\RequirePackage{cleveref}
\RequirePackage{array}
\RequirePackage{graphicx}
\RequirePackage{xcolor}
\RequirePackage{enumitem}
\RequirePackage{tabularx}

\date{}

\theoremstyle{definition}
\newtheorem{definition}{Definition}

\theoremstyle{plain}
\newtheorem{theorem}{Theorem}
\newtheorem{lemma}{Lemma}
\newtheorem{proposition}{Proposition}
\newtheorem{corollary}{Corollary}

\newtheoremstyle{bfnote}
{}{}
{}{}
{\bfseries}{.}
{ }
{\thmname{#1}\thmnumber{ #2}\thmnote{\textnormal{ (#3)}}}
\theoremstyle{bfnote}

\newcommand{\Omit}[1]{}

\usepackage{mathtools}
\usepackage{theoremref}
\usepackage[all,cmtip]{xy}
\usepackage{tikzpeople}
\usepackage{xcolor}

\def\nd{3.0em} 	
\def\ra{0.8em}

\newlength{\wordlength}
\newcommand{\mathwordbox}[3][c]{\settowidth{\wordlength}{$#3$}\makebox[\wordlength][#1]{$#2$}}

\newcommand{\R}{\mathrel{R}}
\newcommand{\RP}{\mathrel{P}}
\newcommand{\I}{\mathrel{I}}
\newcommand{\RAC}{\mathrel{R^A_C}}
\newcommand{\PAC}{\mathrel{P^A_C}}
\newcommand{\IAC}{\mathrel{I^A_C}}
\newcommand{\RA}{\mathrel{R^A}}
\newcommand{\PA}{\mathrel{P^A}}
\newcommand{\RBC}{\mathrel{R^B_C}}
\newcommand{\PBC}{\mathrel{P^B_C}}
\newcommand{\RB}{\mathrel{R^B}}
\newcommand{\PB}{\mathrel{P^B}}
\newcommand{\base}{\mathrel{\overline{P}}}
\newcommand{\baseC}{\mathrel{\overline{P}_C}}
\newcommand{\baseRC}{\mathrel{\overline{R}_C}}
\newcommand{\baseR}{\mathrel{\overline{R}}}

\newcommand{\baseI}{\mathrel{\overline{I}}}
\newcommand{\splitcycle}{\mathit{SC}}

\renewcommand{\epsilon}{\varepsilon}

\title{On Locally Rationalizable\\ Social Choice Functions}

\author{Felix Brandt \quad Chris Dong\\Technische Universit\"at M\"unchen
}

\begin{document}

\maketitle

\begin{abstract}
We consider a notion of rationalizability, where the rationalizing relation may depend on the set of feasible alternatives. More precisely, we say that a choice function is locally rationalizable if it is rationalized by a \emph{family} of rationalizing relations such that a strict preference between two alternatives in some feasible set is preserved when removing other alternatives.
\citet{Tyso08a} has shown that a choice function is locally rationalizable if and only if it satisfies Sen's~$\gamma$. We expand the theory of local rationalizability by proposing a natural strengthening of $\gamma$ that precisely characterizes local rationalizability via PIP-transitive relations and by introducing the $\gamma$-hull of a choice function as its finest coarsening that satisfies $\gamma$.
Local rationalizability permits a unified perspective on social choice functions that satisfy $\gamma$, including classic ones such as the \emph{top cycle} and the \emph{uncovered set} as well as new ones such as \emph{two-stage majoritarian choice} and \emph{split cycle}. 
We give simple axiomatic characterizations of some of these using local rationalizability and propose systematic procedures to define social choice functions that satisfy~$\gamma$.
\end{abstract}

\noindent\textbf{Keywords}: Choice theory, Rationalizability, Consistency, Social Choice

\section{Introduction}
\label{sec:intro}

Rational choice theory studies functions that associate each feasible set of alternatives with a nonempty subset of chosen alternatives. Typically, the choices of rational agents are supposed to be rationalizable in the sense that there exists an underlying binary preference relation on the set of all alternatives such that the choice from any feasible subset consists precisely of the elements that are maximal according to this relation. When the number of feasible alternatives is finite, a relation can rationalize a choice function if and only if its strict part is acyclic. Stronger notions of rationalizability can be obtained by furthermore requiring the rationalizing relation to be quasi-transitive or transitive.

As \citet{Arro51a} and others have famously demonstrated, these notions of rationalizable choice can hardly be sustained in the context of \emph{social} choice. There are several impossibility results showing that even the weakest form of rationalizability is incompatible with mild conditions that are considered indispensable for collective decision-making \citep[for an overview of the extensive literature, see][]{Kell78a,Sen77a,Sen86a,Schw86a,CaKe02a}.
\citet{Sen71a} has characterized rationalizable choice functions using two consistency conditions that relate choices from varying feasible sets to each other, namely conditions $\alpha$ and~$\gamma$ (aka contraction and expansion consistency). $\alpha$ demands that every chosen alternative remains chosen in every feasible subset in which it is contained, and~$\gamma$ requires that every alternative that is chosen from two feasible sets will also be chosen from the union of both sets.
While $\alpha$ has been identified as the main culprit of the impossibility results in social choice \citep{Sen77a,Sen86a},~$\gamma$ and even strengthenings of~$\gamma$ such as $\beta$ and $\beta^+$ appear to be much less harmful. As a matter of fact, two majoritarian social choice functions---the top cycle and the uncovered set---are known to satisfy $\beta^+$ and~$\gamma$, respectively.

The goal of this paper is to improve our understanding of choice functions that satisfy~$\gamma$ and its stronger siblings. For one, this opens an avenue for escaping from the notorious Arrovian impossibilities. Secondly,~$\gamma$ and related consistency conditions are natural and appealing in their own right.
Only few social choice functions are known to satisfy~$\gamma$. Among them are variants of the \emph{top cycle} and the \emph{uncovered set} \citep[see, e.g.,][]{Bord83a}, and two recently proposed functions called \emph{split cycle} \citep{HoPa20a} and \emph{two-stage majoritarian choice} \citep{HoSp21a}. Our results illuminate the similarities of these functions and enable the definition of new 
rules that satisfy~$\gamma$.

The key to get a grip on choice functions that satisfy~$\gamma$ is a weakened notion of rationalizability, where the rationalizing relation may vary depending on the feasible set. Clearly, without further restrictions, any choice function $C$ would be rationalizable in this relaxed setting by letting the rationalizing relation $R^A$ for feasible set $A$ be the relation in which each element of $C(A)$ is strictly preferred to each element of $A\setminus C(A)$. We therefore impose the following natural restriction on families of rationalizing relations: for two feasible sets $A,B$ with $B\subseteq A$ and two alternatives $x,y\in B$ with $x \PA y$, we demand that $x \PB y$. In other words, a strict preference between two alternatives is preserved when reducing the feasible set. Since rationalizing relations are complete, this is equivalent to demanding that $x \RB y$ implies $x\RA y$. We then say that a choice function $C$ is locally rationalizable if there is a family of acyclic and complete relations $(\RA)_A$ that rationalize $C$ in this sense. 

\citet{Tyso08a} has established illuminating connections between local rationalizability and choice consistency conditions. 
\begin{itemize}
\item A choice function satisfies~$\gamma$ if and only if it is locally rationalizable (\Cref{thm:Gamma}).
\item A choice function satisfies~$\gamma$ and $\epsilon^+$ if and only if it is locally rationalized by a family of quasi-transitive relations (\Cref{thm:GammaAizerman}).
\item A choice function satisfies $\beta^+$ if and only if it is locally rationalized by a family of transitive relations (\Cref{thm:BetaPlus}).
\end{itemize}
As corollaries of these results, one obtains classic characterizations of rationalizability, quasi-transitive rationalizability, and transitive rationalizability by \citet{Sen71a}, \citet{Schw76a}, and \citet{Arro59a}.

We extend these results by giving a characterization of choice functions that are locally rationalizable via PIP-transitive relations, an intermediate transitivity notion proposed by \citet{Schw76a} that lies in between quasi-transitivity and full transitivity. Such functions are characterized by $\gamma^+$, a new natural strengthening of~$\gamma$, which requires that for all feasible sets $A,B$, $C(A)\subseteq C(A\cup B)$ or $C(B)\subseteq C(A\cup B)$. 

All of the above results make reference to families of local revealed preference relations that are sandwiched in between the base relation and the revealed preference relation of the choice function at hand. We introduce the \emph{$\gamma$-hull} of a given choice function $C$ as the unique finest coarsening of $C$ that satisfies $\gamma$. In \Cref{thm:gammahull}, we prove that the $\gamma$-hull of $C$ is locally rationalized by the family of local revealed preference relations of $C$. 

Local rationalizability allows us to give a simple characterization of the top cycle based on an observation by \citet{Bord76a}: 
the top cycle is the finest choice function satisfying transitive local rationalizability.
We also give a new characterization of Gillies' uncovered set as the finest choice function satisfying quasi-transitive local rationalizability and a condition we call weak idempotency ($C(C(A))= C(A)$ whenever $\lvert C(A) \rvert = 2$). By introducing further technical axioms, we obtain characterizations of other variants of the uncovered set due to Bordes, McKelvey, and Duggan. We also use local rationalizability to characterize the split cycle rule.

The notion of local rationalizability enables a unified perspective on social choice functions that satisfy~$\gamma$. The idea is to view weak majority rule as a directed graph on the set of alternatives, break majority cycles in the feasible set by removing edges of this graph, and then return the maximal elements of the feasible set according to the resulting acyclic graph. Different rules specifying which strict edges ought to be removed lead to different social choice functions. 
\begin{itemize}
	\item \emph{Top cycle}: Remove all edges that lie on cycles.
	\item \emph{Untrapped set}: Remove all edges that lie on strict cycles.
	\item \emph{Uncovered set}: Remove edges that lie on three-cycles.\footnote{Which edges are to be deleted from each three-cycle depends on which variant of the uncovered set is considered. For example, for Duggan's deep uncovered set, all strict edges are deleted, whereas for McKelvey's uncovered set, only strict edges that lie on three-cycles in which at least two edges are strict are deleted. In the absence of majority ties, all variants coincide.}
	\item \emph{Split cycle}: Remove all edges with minimal majority margins from each cycle.
	\item \emph{Two-stage majoritarian choice}: Fix a strict order over all alternatives. In each cycle, remove the edge departing from the maximal alternative w.r.t. the order.
\end{itemize}
Moreover, for any given choice function $C$, we define a new choice function, the \emph{$\gamma$-core of $C$}, which satisfies $\gamma$. Here, in each cycle, all outgoing edges from alternatives that are selected by $C$ are removed from the majority graph.

Note that the rule specifying which edges should be deleted must only depend on the alternatives contained in the cycle. Otherwise, the resulting social choice function is not guaranteed to satisfy~$\gamma$. This is, for example, the case for Kemeny's rule, ranked pairs, and Schulze's rule, which share the same idea of taking maximal elements after breaking majority cycles \citep[see, e.g.,][for definitions]{FHN15a}. However, the cycles are not broken locally, which results in functions that violate~$\gamma$.

The remainder of the paper is structured as follows. The choice theory model, as well as standard notions of rationalizability and consistency, are introduced in \Cref{sec:Preliminaries}. This is then generalized to local rationalizability and families of local revealed preference relations in \Cref{sec:localrat}.
\Cref{sec:localratChars} shows that each of three notions of local rationalizability is equivalent to an expansion consistency condition, with local revealed preference relations playing a key role.
In \Cref{sec:results}, we present two new additions to the local rationalizability framework by proposing and characterizing a natural strengthening of~$\gamma$, as well as introducing the concept of the~$\gamma$-hull, defined as the finest coarsening of a choice function that satisfies~$\gamma$.
Finally, we discuss various social choice functions that can be described using local rationalizability in \Cref{sec:socialchoice} and introduce the~$\gamma$-core, a systematic procedure to generate social choice functions satisfying~$\gamma$.

\section{Preliminaries}
\label{sec:Preliminaries}

Let $U$ be a nonempty universe of alternatives.  
In this paper, every nonempty and finite subset of $U$ is a \emph{feasible set}. The set of all feasible sets is denoted by $\mathcal{U}$.
A \emph{choice function} $C$ maps each feasible set $A\in \mathcal U$ to a nonempty subset of $A$.
The set of all choice functions is denoted by $\mathcal{C}$.
For two choice functions $C, \widehat{C}$, we say that $C$ is a \emph{refinement} of $\widehat{C}$ if $C(A)\subseteq \widehat{C}(A)$ for all $A\in\mathcal{U}$. Analogously,  $\widehat{C}$ is a \emph{coarsening} of $C$. For a collection of choice functions $\widehat{\mathcal{C}} \subseteq \mathcal{C}$, we say that $C\in \widehat{\mathcal{C}}$ is the \emph{finest} choice function in $\widehat{\mathcal{C}}$ if $C$ is a refinement of every $\widehat{C}\in\widehat{\mathcal{C}}$.

If $R$ is a relation, we denote its strict part by $P$ and its symmetric part by $I$.
$R$ is \emph{transitive} (\emph{quasi-transitive}, \emph{acyclic}, resp.) if for all $x,y,z,x_1,\dots,x_k\in U$, 
\begin{align}
	x\R y \text{ and }y\R z &\text{ implies }x\R z\text, \tag{transitivity}\\
	x\RP y \text{ and }y\RP z &\text{ implies }x\RP z\text,\tag{quasi-transitivity}\\
	x_1\RP x_2,\dots,x_{k-1} \RP x_k &\text{ implies }x_1 \R x_k\text.\tag{acyclicity}
\end{align}
Transitivity implies quasi-transitivity, which implies acyclicity.

Let $R\subseteq U\times U$ be a relation on $U$ and $A$ a feasible set. 
The set of \emph{maximal elements} in $A$ with respect to $R$ is defined by
\[\max_R A=\{ x\in A \colon y\RP x \text{ for no } y\in A\}\text.\]
Note that $\max_R A$ is nonempty for all $A\in \mathcal{U}$ if and only if $R$ is acyclic.

A choice function $C$ is \emph{rationalizable} (\emph{quasi-transitively rationalizable},  \emph{transitively rationalizable}, resp.) if there is an acyclic (quasi-transitive, transitive, resp.) and complete relation $R$ on $U$ such that for all feasible sets $A$,
\[
	C(A)= \max_{R}{A}\text.\tag{rationalizability}
\]
In this case we say that $C$ is rationalized by $R$.
Two natural candidates for the rationalizing relation are the \emph{base relation} $\overline R_C$~\citep{Herz73a} and the \emph{revealed preference relation} $R_C$~\citep{Samu38a,Hout50a}, which, for all alternatives~$x$ and~$y$, are given by
\begin{align}
	x\mathrel{\overline{R}_C}y &\quad\text{iff}\quad x\in C(\{x,y\})\text,\tag{base relation}\\
	x\mathrel{R_C}y &\quad\text{iff}\quad x\in C(A)\text{ for some $A\in\mathcal{U}$ with $y\in A$.}\tag{revealed preference relation}
\end{align}
The revealed preference relation relates~$x$ to~$y$ if $x$ is chosen in the presence of~$y$ and possibly other alternatives, whereas the base relation only relates~$x$ to~$y$ if $x$ is chosen in the exclusive presence of~$y$. Both the base relation and the revealed preference relation are complete by definition. The revealed preference relation is furthermore guaranteed to be acyclic. 
Whenever $C$ is rationalizable, the base relation and the revealed preference relation coincide and rationalize $C$.

We now define four choice consistency conditions. For all feasible sets~$A$ and~$B$,
\begin{align*}
(\alpha): \qquad &\text{if }B\subseteq A \text{, then } C(A)\cap B\subseteq C(B)\qquad &\text{\citep{Cher54a}}\\
(\gamma): \qquad &C(A)\cap C(B) \subseteq C(A\cup B)\qquad &\text{\citep{Sen71a}}\\
(\epsilon^+): \qquad &\text{if } C(A)\subseteq B\subseteq A \text{, then } C(B)\subseteq C(A)\qquad & \text{\citep{Bord83a}}\\
(\beta^+): \qquad &\text{if }B\subseteq A \text{ and } C(A)\cap B \neq \emptyset \text{, then } C(B)\subseteq C(A)\qquad &\text{\citep{Bord76a}}
\end{align*}

\emph{Contraction consistency} conditions specify under which circumstances an alternative chosen from some feasible set is still chosen from a feasible \emph{subset}. Similarly, \emph{expansion consistency} conditions specify under which circumstances a chosen alternative is chosen from a feasible \emph{superset}. While $\alpha$ is a contraction consistency condition,~$\gamma$, $\epsilon^+$ and $\beta^+$ are expansion consistency conditions.\footnote{$\beta^+$ was first considered by \citet{Bord76a}. It is stronger than Sen's $\beta$, but equivalent to $\beta$ in the presence of $\alpha$.
\citet{Bord83a} introduced $\epsilon^+$ as a strengthening of $\epsilon$ by \citet{BBKS76a}. It is also known as Aïzerman \citep{Moul86a}, 
the weak superset property \citep{Bran11b}, 
and $\hat{\alpha}_\subseteq$ \citep{BBH16a}.
$\epsilon^+$ can be derived as the ``expansion part'' of Postulate $5^*$ introduced by \citet{Cher54a}.
}
$\beta^+$ is typically seen as the strongest (non-trivial) expansion consistency condition. It implies both~$\gamma$ and $\epsilon^+$, which are logically independent of each other. Choice theory has identified a number of fundamental relationships between consistency and rationalizability. Let $C$ be an arbitrary choice function.

\begin{align*}
\text{$C$ satisfies $\alpha$ and~$\gamma$} & \quad\text{iff}\quad\text{$C$ is rationalizable.} &\text{\citep{Sen71a}}\\
\text{$C$ satisfies $\alpha$,~$\gamma$, and $\epsilon^+$} & \quad\text{iff}\quad\text{$C$ is quasi-transitively rationalizable.} &  \text{\citep{Schw76a}}\\
\text{$C$ satisfies $\alpha$ and $\beta^+$} & \quad\text{iff}\quad\text{$C$ is transitively rationalizable.} & \text{\citep{Bord76a}}
\end{align*}

\section{Local Rationalizability and Local Revealed Preference}
\label{sec:localrat}

In this section, we define a weakening of rationalizability, which allows the rationalizing relation to vary depending on the feasible set. Formally, there is a family of rationalizing relations $(R^A)_{A\in\mathcal{U}}$, one for each feasible set, such that $R^B \subseteq R^A$ whenever $B\subseteq A$. We also define, for every choice function $C$, a corresponding family of revealed preference relations $(R^A_C)_{A\in\mathcal{U}}$ by restricting the witness for a preference in $R^A_C$ to subsets of $A$.

\citeauthor{Tyso08a} refers to local rationalizability as \emph{admitting a satisficing representation} while
\citeauthor{Dugg19a} uses the term \emph{pseudo-rationalizability}. However, the latter is already used by \citet{Moul80a} for a different notion in choice theory. We think the term `local' aptly describes that the rationalizing relation may depend on the feasible set.

\begin{definition}[Local rationalizability]
	\label{Def:UpRat}
	A choice function $C$ is \emph{locally rationalizable} if there is a family of relations $(R^A)_{A\in\mathcal{U}}$ such that for all feasible sets $A$,
	\begin{enumerate}[label=\textit{(\roman*)}]
		\item  \label{it:acyclicComplete}
				$R^A\subseteq A\times A$ is acyclic and complete,
		\item  \label{it:maximalChoice}
				$C(A)= \max_{R^A}{A}$, and
		\item \label{it:inclusion}
				$R^B \subseteq R^A$ for all feasible sets $B\subseteq A$.
	\end{enumerate}
	In this case, we say that $C$ is \emph{locally rationalized by} $(R^A)_A$.
\end{definition}

Conditions \ref{it:acyclicComplete} and \ref{it:maximalChoice} are analogous to the case of standard rationalizability. 
Condition \ref{it:inclusion} implies that a strict preference between two alternatives in some feasible set is preserved when removing other alternatives. Since all relations $R^A$ are complete and hence $x \RA y$ if and only if not $y \PA x$, Condition \ref{it:inclusion} can be reformulated as
\begin{itemize}\label{cond:3prime}
	\item [\textit{(iii')}] $P^A\cap {(B\times B)}\subseteq P^B$ for all feasible sets $B\subseteq A$.
\end{itemize}
The crucial difference to classic rationalizability is that a strict preference between two alternatives can be revoked by introducing new alternatives.
Without Condition \ref{it:inclusion}, all choice functions would satisfy local rationalizability. Standard rationalizability via a global relation $R$ implies local rationalizability by letting $R^A = R|_A$ for each feasible set~$A$.\footnote{Remarkably, if we replace the subset in Condition \ref{it:inclusion} with a superset, then we obtain a dual definition that characterizes $\alpha$ for countable $U$. However, the results for local rationalizability can be structured much more elegantly than the corresponding theory for $\alpha$. For example, there is not always a unique inclusion minimal family of rationalizing relations, and quasi-transitivity does not add anything to acyclicity in this context.}

At first glance, it seems difficult to verify whether a choice function is locally rationalizable because there is a very large number of families of potentially rationalizing relations. It turns out that a rather natural candidate for such a family is obtained by extending the concept of revealed preference to families of relations that depend on the feasible set.
\begin{definition}[Local Revealed Preference Relations]
	\label{Def:LRP}
	Let $C$ be a choice function and $A$ a feasible set with $x,y\in A$. 
	We write $x \mathrel{R_C^A} y$ if and only if there is some feasible set $B\subseteq A$ with $x\in C(B)$ and $y\in B$. We call $R_C^A$ the \emph{local revealed preference relation} on $A$ and $(R_C^A)_{A\in \mathcal U}$ the \emph{family of local revealed preference relations}.
\end{definition}

The difference to classic revealed preference is that we now locally restrict our witness $B$ to be a subset of $A$, while the classic notion allows for arbitrary witnesses.

It follows from \Cref{Def:LRP} that a \emph{strict} local revealed preference between alternatives $x$ and $y$ in $A$ holds if and only if $y$ is not chosen from any subset that contains $x$, i.e.,
\begin{equation}\label{equ:strict_lrp}
x \PAC y \quad\text{iff}\quad y \notin C(B) \text{ for all } B\subseteq A \text{ with } x,y \in B\text.
\end{equation}
 
It is easily seen that the base relation of $C$ is equivalent to the local revealed preference relations of two-element sets, that is, for all $x,y\in U$, $x\baseRC y$ if and only if $x\mathrel{R^{\{x,y\}}_C} y$. 
Moreover, all local revealed preference relations are sandwiched in between the base relation and the revealed preference relation, i.e., for all $A\in \mathcal{U}$,
\begin{equation}\label{equ:lrp_inclusion}
\overline{R}_C\cap(A\times A) \quad\subseteq\quad R^A_C \quad\subseteq\quad R_C\text.
\end{equation}
The first inclusion holds because every two-element feasible set can serve as a witness, and the second follows from $R^U_C =R_C$.

The following lemma shows that the family of local revealed preference relations already satisfies all but one of the conditions required for local rationalizability. 

\begin{lemma}
	\label{lem:LRP}
	Let $C$ be a choice function. Then, $(R_C^A)_A$ satisfies Condition \ref{it:acyclicComplete} and \ref{it:inclusion} of \Cref{Def:UpRat}.\footnote{On top of that, the inclusion from left to right in Condition \ref{it:maximalChoice} ($C(A)\subseteq \max_{R^A}{A}$) is satisfied.}
\end{lemma}
\begin{proof}
Completeness of $R_C^A$ follows directly from \Cref{equ:lrp_inclusion} and the completeness of the base relation.

	To show acyclicity of $R_C^A$, 
	let $A$ be a feasible set and $x_1,\dots, x_k\in A$ such that $x_i \PAC x_{i+1}$ for all $0<i<k$. We know that $B\coloneqq\{x_i\colon1\leq i\leq k\}\subseteq A$. For each $y\neq x_1$, there is some element $x\in B$ with $x\PAC y$. Applying \Cref{equ:strict_lrp} to $B\subseteq A$, it follows that $y\notin C(B)$. By nonemptiness we conclude $\{x_1\}=C(B)$. Hence, $B$ is a witness for $x_1\RAC x_k$.
	
	To show Condition \ref{it:inclusion} of \Cref{Def:UpRat},
	let $A,B$ be feasible sets such that $B\subseteq A$ and $x,y\in B$ such that $x \RBC y$. 
	Now let $D\subseteq B$ be the witness containing $y$ with $x\in C(D)$. 
	Then of course $y\in D\subseteq A$ and still $x\in C(D)$. 
	By definition, this implies $x\RAC y$. Hence, $R^B_C\subseteq R^A_C$.
\end{proof}

The following example illustrates the difference between local rationalizability and rationalizability.
\paragraph{Example.}
	Let $U=\{x,y,z\}$ and $C$ be defined on all non-singleton subsets as follows.
	\begin{center}
		\begin{tabular}{ c c  }
			$A$ & $C(A)$ \\ 
			\hline
			$\{x,y,z\}$ & $\{x,y\}$\\
			$\{x,y\}$ & $\{x\}$\\
			$\{y,z\}$ & $\{y\}$\\
			$\{x,z\}$ & $\{x\}$\\
		\end{tabular}
	\end{center}
	For $B\coloneqq\{x,y\}$, we have $x \PBC y$. On the other hand, we have $y\mathrel{R_C} x$, since $y$ is chosen from $A\coloneqq\{x,y,z\}$.
	Hence, we have $y\in \max_{R_C} B$, but $y\notin \max_{R_C^B} B$.
	Despite $x \PBC y$, we still have $y \mathrel{I^A_C} x$ and thus $y\in \max_{R_C^A} A$.	
	In summary, $C$ is not rationalizable but does satisfy local rationalizability. 
\newline

It is easily seen that when assuming $\alpha$, all revealed preference relations are restrictions of the same relation to the given feasible set, and our generalized notion of revealed preference coincides with the classic one.

\begin{lemma}
	\label{lem:RevPrefandAlpha}
	Let $C$ be a choice function that satisfies $\alpha$. Then, $R^A_C=R_C\cap (A\times A)$ for all feasible sets $A$.
\end{lemma}
\begin{proof}
	By definition we have $R^A_C\subseteq R_C\cap (A\times A)$.
	Now let $x \mathrel{R_C} y$ with $x,y\in A$. Then there is some witness $D\in \mathcal{U}$ with $x\in C(D)$, $y\in D$. By $\alpha$, we have $x\in C(\{x,y\})$. Since $\{x,y\}\subseteq A$, this yields $x\RAC y$.
\end{proof}

\section{Rationalizability and Consistency}
\label{sec:localratChars}

The unifying theme of the following propositions is that expansion consistency is deeply intertwined with local rationalizability, even without imposing contraction consistency. Three classic characterizations of rationalizability can be obtained as corollaries. 

In contrast to our original belief, we are not the first to consider local rationalizability and discover the preceding propositions. 
Propositions \ref{thm:Gamma}, \ref{thm:GammaAizerman}, and \ref{thm:BetaPlus} were  shown by \citet[][Thms.~6, 7, and 2]{Tyso08a} and then again by \citet[][Thms.~7, 8, and 9]{Dugg19a}, apparently unaware of \citeauthor{Tyso08a}'s work. The core of \Cref{thm:Gamma} goes back to even earlier work by \citet[][Thm.~4]{Aize85a}.
We provide proofs of these statements in order to have a self-contained presentation with consistent notation before we extend these results in Sections \ref{sec:results} and \ref{sec:socialchoice}.

We start with the most central characterization, showing the equivalence between local rationalizability and~$\gamma$. This equivalence will later be leveraged to construct (social) choice functions that satisfy~$\gamma$.

\begin{proposition}[\citealp{Tyso08a}]
	\label{thm:Gamma}
	A choice function satisfies~$\gamma$ if and only if it is locally rationalizable. Moreover, any such function is rationalized by its family of local revealed preference relations. 
\end{proposition}

\begin{proof}
	Let $C$ be a choice function.
	For the direction from right to left, assume that $(R^A)_A$ locally rationalizes $C$ and let $x \in C(A)\cap C(B)$.
	Now consider an arbitrary $y\in A\cup B$. If $y\in A$, then $x\in C(A)$ implies $x \RA y$. Otherwise, $y \in B$,  and $x \RB y$ because $x\in C(B)$. In both cases, we can apply condition \ref{it:inclusion} of 	\Cref{Def:UpRat} to obtain $x\mathrel{R^{A\cup B}}y$. Since $y$ was arbitrary, we have $x \R^{A\cup B} y$ for all $y\in A\cup B$. Hence, $x\in \max_{R^{A\cup B}} A\cup B=C(A\cup B)$.
	
	For the direction from left to right, we will leverage the family of local revealed preference relations $(R^A_C)_A$. Assume that $C$ satisfies~$\gamma$ and let $A$ be a feasible set. By \Cref{lem:LRP} we only need to show that $C(A)=\max_{R_C^A} A$.
	For the inclusion from left to right, 
	let $x\in C(A)$. We then have by definition that $x \RAC y$ for all $y\in A$ and consequently that $x\in \max_{R_C^A} A$. 
	For the inclusion from right to left, let $x$ be maximal in A. Now, let $y\in A$ be given. 
	By maximality of $x$ and completeness of $\RAC$, we know that $x\RAC y$. 
	Hence, there is some $B_y\subseteq A$ such that $x\in C(B)$ and $y\in B_y$.
	Since $A$ is finite and $y$ was arbitrary, we can repeatedly apply~$\gamma$ to obtain  $x\in C(\cup_{y\in A} B_y)=C(A)$.
\end{proof}

The family of locally revealed preference relations is not the only family of relations locally rationalizing a choice function. This is hardly surprising because locally rationalizing relations for larger feasible sets may contain additional indifferences that are irrelevant for the set of maximal elements. Note that these indifferences are not in conflict with Condition \ref{it:inclusion} of \Cref{Def:UpRat}.

It can be shown that the family of local revealed preference relations is the \emph{finest} family of locally rationalizing relations in the following sense: we say that $(R^A)_A$ \emph{is finer than} $(\tilde{R}^A)_A$, if $R^A\subseteq \tilde{R}^A$ for all $A\in \mathcal{U}$. Thus, when restricting attention to locally rationalizing relations that are minimal in this sense, uniqueness of local revealed preference is retained.

We can use the characterization of~$\gamma$ to obtain a classic result of Sen as a corollary.
Since the latter involves $\alpha$, we need the additional observation that rationalizability implies $\alpha$. To see this, let $R$ be a relation that rationalizes choice function $C$ and $x\in C(A)\cap B$. Then, for all $y\in B$, $x\R y$. Hence, $x\in C(B)$.
\begin{lemma}
	\label{lem:RatImpliesAlpha}
	Every rationalizable choice function satisfies $\alpha$.
\end{lemma}

We now obtain Sen's characterization of rationalizable choice functions as follows. For any rationalizable choice function $C$, \Cref{thm:Gamma} and \Cref{lem:RatImpliesAlpha} imply that $C$ satisfies $\alpha$ and~$\gamma$. Conversely, if $C$ satisfies $\alpha$ and~$\gamma$, \Cref{thm:Gamma} and \Cref{lem:RevPrefandAlpha} imply that $R_C$ rationalizes $C$.

\begin{corollary} [\citealp{Sen71a}]
	\label{thm:Sen}
		A choice function satisfies $\alpha$ and~$\gamma$ if and only if it is rationalizable. Moreover, any such function is rationalized by its revealed preference relation.
\end{corollary}

As it turns out, requiring that all locally rationalizing relations are quasi-transitive leads to choice functions that not only satisfy~$\gamma$ but also $\epsilon^+$.

\begin{proposition}[\citealp{Tyso08a}]
	\label{thm:GammaAizerman}
	A choice function satisfies~$\gamma$ and $\epsilon^+$ if and only if it is locally rationalized by a family of quasi-transitive relations.
\end{proposition}
\begin{proof}
	For the direction from left to right, let $C$ be a choice function satisfying~$\gamma$ and $\epsilon^+$. 
	By \Cref{thm:Gamma}, we already know that the local revealed preference relations locally rationalize $C$. 
	In addition, we now show that each of these relations is quasi-transitive.
	Let $A$ be a feasible set and $x,y,z\in A$ such that $x \PAC y$ and $y \PAC z$. It needs to be shown that $x \PAC z$ or, in other words, that $z$ is not chosen from any subset of $A$ that contains $x$.
	For this, consider an arbitrary $B\subseteq A$ with $x,z\in B$.
	Note that if $y\in B$, we directly obtain $z\notin C(B)$ from \Cref{Def:LRP}.
	Otherwise, set $B_y\coloneqq B\cup \{y\}$.
	By $x,y\in B_y$, we have $y\notin C(B_y)$, which implies $C(B_y)\subseteq B \subseteq B_y$. It follows from $\epsilon^+$ that $C(B)\subseteq C(B_y)$.
	Since $B_y\subseteq A$ and $y,z\in B_y$, we have $z\notin C(B_y)$.
	Hence we can conclude $z\notin C(B)$.
	Since $B$ was arbitrary, we obtain $x\mathrel{P^A_C} z$.
	
	For the other direction, let $(R^A)_A$ be a family of quasi-transitive preference relations which locally rationalizes choice function $C$.
	We already know that $C$ satisfies~$\gamma$ by \Cref{thm:Gamma}.
	For $\epsilon^+$, let $A,B$ be feasible sets such that $C(A)\subseteq B\subseteq A$. It needs to be shown that $C(B)\subseteq C(A)$.	
	In other words, it suffices to show that for any $z\in B\setminus C(A)$, it holds that $z\notin C(B)$. 
	Since $z\notin C(A)$, there must be some $x_1\in A$, such that $x_1\PA z$.
	If $x_1\in B$, then by local rationalizability $x_1\PB z$ and hence $z\notin C(B)$. Otherwise, by $C(A)\subseteq B$, it must be that $x_1\notin C(A)$. 
	Hence there must be $x_2\in A$ with $x_2 \PA x_1$. 
	It follows from the quasi-transitivity of $R^A$ that $x_2 \PA z$.
	Using induction, quasi-transitivity of $R^A$, and finiteness of $A$, there must eventually be some $x_\ell\in B$ with $x_\ell \PA z$.
	Hence, $x_\ell \PB z$ and thus $z\notin C(B)$.
\end{proof}

It follows from \Cref{thm:Gamma} that every quasi-transitively locally rationalizable choice function is rationalized by its family of local revealed preference relations. Additionally, this constitutes the finest family of quasi-transitive locally rationalizing relations. Note that other families of rationalizing relations may not be quasi-transitive.

Again, we obtain a classic characterization as a corollary. Let $C$ be a quasi-transitively rationalizable choice function. By \Cref{thm:GammaAizerman} and \Cref{lem:RatImpliesAlpha}, $C$ satisfies $\alpha$,~$\gamma$ and $\epsilon^+$. For the converse direction, let $C$ satisfy $\alpha$,~$\gamma$ and $\epsilon^+$. Then by \Cref{thm:GammaAizerman} and \Cref{lem:RevPrefandAlpha} we have that $R_C$ is quasi-transitive and rationalizes $C$.

\begin{corollary}[\citealp{Schw76a}]
	\label{thm:Schwartz}
	A choice function satisfies $\alpha$,~$\gamma$ and $\epsilon^+$ if and only if it is quasi-transitively rationalizable.
\end{corollary}

It turns out that choice functions that are locally rationalized via transitive relations are characterized by the strongest expansion consistency condition $\beta^+$.

\begin{proposition}[\citealp{Tyso08a}]
	\label{thm:BetaPlus}
	A choice function satisfies $\beta^+$ if and only if it is locally rationalized by a family of transitive relations.
\end{proposition}
\begin{proof} 
	For the direction from left to right, let $C$ be a choice function satisfying $\beta^+$.
	Since $\beta^+$ implies~$\gamma$, \Cref{thm:Gamma} applies and $C$ must be locally rationalized by its family of local revealed preference relations.
	Now, let $A$ be a feasible set.
	We show that $R^A_C$ is transitive.
	Let $x,y,z$ be given such that there are feasible sets $B_1,B_2\subseteq A$ with $x\in C(B_1)$, $y\in C(B_2)$, $y\in B_1$, and $z\in B_2$. 
	It needs to be shown that there is some feasible $D\subseteq A$ such that $D$ contains $z$ and $x$ is chosen from $D$.
	Let $D\coloneqq B_1\cup B_2$. Observe that by $\beta^+$ and $B_1\subseteq D$, if $C(D)\cap B_1 \neq \emptyset$, then $x\in C(D)$. 
	If $C(D)\cap B_2=\emptyset$, this follows from nonemptiness of choice sets. Otherwise, $\beta^+$ applied on $B_2 \subseteq D$ and $y\in C(B_2)$ implies $y\in C(D)$.
	Since $y\in B_1$, this concludes the first half of the proof.
	
	For the direction from right to left, let $(R^A)_A$ locally rationalize $C$ such that all relations are transitive. Now, let $A,B$ be feasible sets such that $B\subseteq A$ and fix $y\in C(A)\cap B\neq \emptyset$. Moreover, let $x\in C(B)$. Then by completeness of $R^B$ and maximality of $x$ in $B$ we obtain $x\RB y$. By local rationalizability, we also have $x \RA y$. Maximality of $y$ in A yields $y\RA z$ for all $z\in A$. Transitivity now implies $x\RA z$ for all $z\in A$. Hence $x$ is maximal in $A$ and it follows that $x\in C(A)$.
\end{proof}

It follows from \Cref{thm:Gamma} that every transitively locally rationalizable choice function is rationalized by its family of local revealed preference relations. Again, this constitutes the finest family of transitive locally rationalizing relations. Note that other families of rationalizing relations may fail to be transitive.
Interestingly, \citet{Bord76a} already observed that $\beta^+$ implies that the revealed preference relation is transitive.

Again, we obtain a classic result as a corollary. 
Let $C$ be a transitively rationalizable choice function. By \Cref{thm:BetaPlus} and \Cref{lem:RatImpliesAlpha}, $C$ satisfies $\alpha$ and $\beta^+$.
Conversely, let $C$ satisfy $\alpha$ and $\beta^+$. Then by \Cref{thm:BetaPlus} and \Cref{lem:RevPrefandAlpha} we have that $R_C$ is transitive and $C$ is rationalized by $R_C$.

\begin{corollary} [\citealp{Arro59a,Bord76a}]
	\label{thm:Bordes}
	A choice function satisfies $\alpha$ and $\beta^+$ if and only if it is transitively rationalizable.
\end{corollary}

\section{Gamma$^+$ and Gamma Hull} 
\label{sec:results}

In this section, we expand the theory of local rationalizability by introducing two new concepts. First, we propose a strengthening of $\gamma$ that precisely characterizes choice functions that are locally rationalizable via PIP-transitive relations. Secondly, we define the $\gamma$-hull of a choice function as its finest coarsening that satisfies $\gamma$. 

\subsection{Gamma$^+$} 

\citet{Schw76a} has proposed a transitivity notion called PIP-transitivity that lies in between quasi-transitivity and full transitivity.
He argues that, in some cases, this notion represents human behavior more accurately than standard transitivity. \citeauthor{Schw86a} states that while transitivity is equivalent to representation by a utility function $u$,
PIP-transitivity is equivalent to representation by a utility function $u$ and a non-negative discriminatory function $\delta$. The idea is that some $a$ is only perceived to be strictly better than some $b$, if the increase in utility is noticeable, which is modeled by $u(a)>u(b)+\delta(b)$.

Let $R$ be a relation on $U$. We say that $R$ is \emph{PIP-transitive} if for all (not necessarily distinct) $x,y,z,w\in U$,
\[
\text{$x \RP y$, $y\I z$, and $z\RP w$ implies $x\RP w$.}\tag{PIP-transitivity}
\]	
	
Graphically, this condition can be represented as below. The wiggly line represents indifference, while the arrows stand for strict preference. The double edge from $x$ to $w$ denotes the consequence of the implication. 
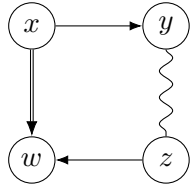
\begin{center}
	\begin{tikzpicture}[node distance=\nd,vertex/.style={circle,draw,minimum size=2*\ra,inner sep=1pt}
		]
		\node[vertex] (x) 	           {$\mathwordbox{x}{x}$};
		\node[vertex] (y) [right=of x] {$\mathwordbox{y}{x}$};
		\node[vertex] (z) [below=of y] {$\mathwordbox{z}{y}$};
		\node[vertex] (w) [left=of z] {$\mathwordbox{w}{z}$};
		\draw [-Latex] (x) to (y);
		\draw [-Latex] (z) to (w);
		\draw [-,decorate,decoration=snake] (y) to (z);
		\draw [-Latex, double] (x) to (w);
	\end{tikzpicture}
\end{center}

We now propose a new expansion consistency condition called $\gamma^+$ and show that it is equivalent to local rationalizability by families of PIP-transitive relations.
\begin{definition}
	\label{Def:GammaPlus}
	A choice function $C$ satisfies $\gamma^+$ if for all feasible sets $A,B$, 
	\[C(A)\subseteq C(A\cup B) \quad\text{or}\quad C(B)\subseteq C(A\cup B)\text.\tag{$\gamma^+$}\]
\end{definition}

Clearly, $\gamma^+$ implies~$\gamma$ since $C(A)\cap C(B)$ is a subset of both $C(A)$ and $C(B)$. Moreover, $\gamma^+$ implies $\epsilon^+$.
To see this, let $C(A)\subseteq B\subseteq A$. Then we use $\gamma^+$ with $\hat B = B$ and $\hat A = A\setminus B$. It must be that $C(\hat B)\subseteq C(A)$, as $C(\hat A)\cap C(A) = \emptyset$.\footnote{$\gamma^+$ is even stronger than the conjunction of~$\gamma$ and $W4$, a condition that was introduced by \citet{Schw76a} for characterizing PIP-transitivity together with $\alpha$. 
}

Analogous to the previous results, we present a characterization of local rationalizability via PIP-transitive relations using $\gamma^+$.
\begin{theorem}
	\label{thm:PIPRat}
	A choice function satisfies $\gamma^+$ if and only if it is locally rationalized by a family of PIP-transitive relations.
\end{theorem}
\begin{proof}
	Let $C$ be a choice function.
	For the direction from left to right, let $C$ satisfy $\gamma^+$. Since $\gamma^+$ implies~$\gamma$, $C$ is locally rationalized by $(R^A_C)_A$. Let $A$ be a feasible set. We now show that the local revealed preference relation $R^A_C$ is PIP-transitive. Let $x\PAC y$, $y\IAC z$ and $z\PAC w$.

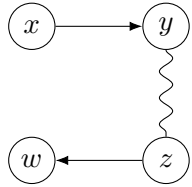
	\begin{center}
		\begin{tikzpicture}[node distance=\nd,vertex/.style={circle,draw,minimum size=2*\ra,inner sep=1pt}
			]
			\node[vertex] (x) 	           {$\mathwordbox{x}{x}$};
			\node[vertex] (y) [right=of x] {$\mathwordbox{y}{x}$};
			\node[vertex] (z) [below=of y] {$\mathwordbox{z}{y}$};
			\node[vertex] (w) [left=of z] {$\mathwordbox{w}{z}$};
			\draw [-Latex] (x) to (y);
			\draw [-Latex] (z) to (w);
			\draw [-,decorate,decoration=snake] (y) to (z);
		\end{tikzpicture}
	\end{center}
	
	It follows from \Cref{equ:strict_lrp} that $w$ cannot be chosen from subsets of $A$ when $z$ is present, and the same holds for $y$ when $x$ is present.
	We now need to show that $x\PAC w$, i.e., $w$ cannot be chosen from subsets of $A$ in the presence of $x$.
	To this end, let $\widehat{A}\subseteq A$ with $x,w\in \widehat{A}$. 
	Further, let $\widehat{B}\subseteq A$ be a witness for $y\IAC z$, i.e., $y \in C(\widehat{B})$ and $z\in \widehat{B}$.
	Now, apply $\gamma^+$ on $\widehat{A}$ and $\widehat{B}$. Since $y\in C(\widehat{B})$ and $x$ is present in $\widehat{A}$, it cannot be that $C(\widehat{B})\subseteq 
	C(\widehat{A}\cup \widehat{B})$.
	Therefore, it has to be that $C(\widehat{A})\subseteq C(\widehat{A}\cup \widehat{B})$. 
	Since $w \notin C(\widehat{A}\cup \widehat{B})$ due to the presence of $z$, we conclude that $w 
	\notin C(\widehat{A})$.
	
	For the other direction, let $(R^A)_A$ locally rationalize $C$ and let all relations be PIP-transitive.
	Let $A,B$ be feasible sets and assume for contradiction that neither $C(A)$ nor $C(B)$ is a subset of $C(A\cup B)$.
	Then, there are $a\in C(A)\setminus C(A\cup B)$ and $b\in C(B)\setminus C(A\cup B)$.	
	Hence there is some $x\in A\cup B$ with $x\mathrel{P^{A\cup B}}a$ and some $y\in A\cup B$ with $y\mathrel{P^{A\cup B}}b$.
	Since $a\in C(A)$, it must be that $x\in B$, and since $b\in C(B)$, we have that $b\RB x$. By local rationalizability, we even have $b\mathrel{R^{A\cup B}} x$.
	There are now two possibilities. 
	First, it could be that $b\mathrel{P^{A\cup B}}x$. By applying quasi-transitivity, we then obtain the following relationships for $R^{A\cup B}$.
	
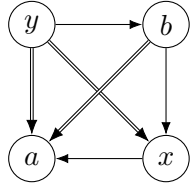
	\begin{center}
		\begin{tikzpicture}[node distance=\nd,vertex/.style={circle,draw,minimum size=2*\ra,inner sep=1pt}
			]
			\node[vertex] (y) 	           {$\mathwordbox{y}{y}$};
			\node[vertex] (b) [right=of y] {$\mathwordbox{b}{y}$};
			\node[vertex] (x) [below=of b] {$\mathwordbox{x}{b}$};
			\node[vertex] (a) [left=of x] {$\mathwordbox{a}{x}$};
			\draw [-Latex] (y) to (b);
			\draw [-Latex] (x) to (a);
			\draw [-Latex] (b) to (x);
			\draw [-Latex, double] (y) to (x);
			\draw [-Latex, double] (b) to (a);
			\draw [-Latex, double] (y) to (a);
		\end{tikzpicture}
	\end{center}
	
	Otherwise, we have $b\mathrel{I^{A\cup B}}x$, which entails the following relationships for $R^{A\cup B}$ since PIP-transitivity can be applied.
	
	
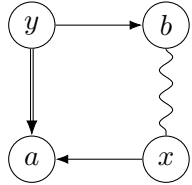
\begin{center}
		\begin{tikzpicture}[node distance=\nd,vertex/.style={circle,draw,minimum size=2*\ra,inner sep=1pt}
			]
			\node[vertex] (y) 	           {$\mathwordbox{y}{y}$};
			\node[vertex] (b) [right=of y] {$\mathwordbox{b}{y}$};
			\node[vertex] (x) [below=of b] {$\mathwordbox{x}{b}$};
			\node[vertex] (a) [left=of x] {$\mathwordbox{a}{x}$};
			\draw [-Latex] (y) to (b);
			\draw [-Latex] (x) to (a);
			\draw [-,decorate,decoration=snake] (b) to (x);
			\draw [-Latex, double] (y) to (a);
		\end{tikzpicture}
	\end{center}
	
	Since $b\in C(B)$, it must be that $y\in A$. Moreover, it follows from $a\in C(A)$ that $a\RA y$. By local rationalizability, we thus get $a\mathrel{R^{A\cup B}}y$, the desired contradiction to $y\mathrel{P^{A\cup B}}a$ in both cases.
\end{proof}

An overview of the characterizations of different notions of local rationalizability via choice consistency conditions is given in \Cref{fig:results}.

\begin{figure}[tbp]
\[\text{\xymatrix @R=1em {\text{transitive local rationalizability}&\text{iff}&&  \beta^+ \ar[d]&	 \\
		\text{PIP-transitive local rationalizability}&\text{iff}&&\gamma^+ \ar[dl]\ar[dr]&	 \\
		\text{quasi-transitive local rationalizability}&\text{iff}& \gamma\ar[d]&    \text{and} & \epsilon^+\\
		\text{(acyclic) local rationalizability}&\text{iff}& \gamma\ar[u]&   & \\			} 
}\]
\caption{Equivalences between local rationalizability and expansion consistency.} 
\label{fig:results}
\end{figure}
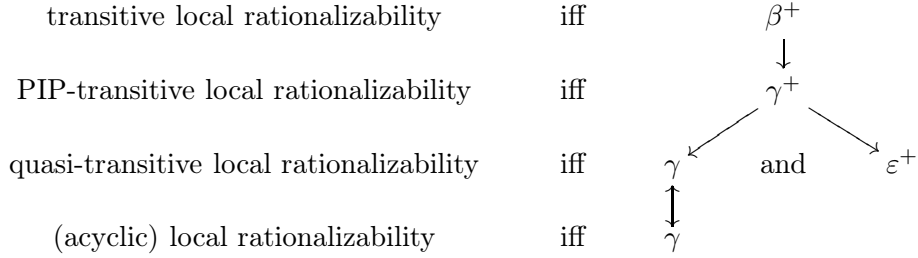

\subsection{Gamma Hull}

Whenever a choice function does not satisfy~$\gamma$, we can `repair' inconsistencies by adding elements to the choice sets. It turns out that there is a unique minimal way of turning any choice function into one that satisfies~$\gamma$.
To this end, first let $\mathcal{G}$ be a nonempty set of choice functions such that $C^*(A)\coloneqq \bigcap_{C \in G}{C(A)}\neq \emptyset$ for all $A\in\mathcal{U}$. Then, $C^*$ is a well-defined choice function that satisfies~$\gamma$ if all choice functions contained in $\mathcal{G}$ do. To see this, let $x \in C^*(A)\cap C^*(B)$ for some $A,B\in \mathcal{U}$. Now, consider an arbitrary $C \in \mathcal{G}$. By definition, $x\in C(A)\cap C(B)$. Since $C$ satisfies~$\gamma$, we have $x\in C(A\cup B)$, and since $C$ was arbitrary, we conclude $x\in C^*(A\cup B)$. Note that $C^*$ is a refinement of all choice functions in $\mathcal{G}$. 
We are now ready to define the~$\gamma$-hull of a choice function $C$, i.e., the unique finest coarsening of $C$ that satisfies~$\gamma$. 

\begin{definition}[$\gamma$-hull]\label{def:gammahull}
	Let $C$ be a choice function and $\mathcal{G}\coloneqq \{\widehat{C}\in \mathcal{C}\colon \text{$\widehat{C}$ satisfies~$\gamma$ and }C\subseteq \widehat{C}\}$. Then, the~$\gamma$-\emph{hull} $C^{\gamma}$ is a choice function given by 
	\begin{equation*}
		C^{\gamma} (A)\coloneqq \bigcap_{\widehat{C} \in \mathcal{G}}{\widehat{C}(A)} \text{ for all $A\in\mathcal{U}$.} 
	\end{equation*}
\end{definition}

While the existence and uniqueness of such a coarsening are already noteworthy, there is a very simple characterization of the~$\gamma$-hull: for any choice function $C$, the~$\gamma$-hull of $C$ is locally rationalized by the family of local revealed preference relations $(R^A_C)_A$. 

\begin{theorem}
	\label{thm:gammahull}
	For any choice function $C$ and feasible set $A$, $C^{\gamma}(A)=\max_{R_C^A} A$.
\end{theorem}
\begin{proof}
	
	For the inclusion from left to right, \Cref{lem:LRP} and \Cref{thm:Gamma} imply that $\max_{R_C^A}$ is locally rationalized by $(R^A_C)_A$ and hence satisfies~$\gamma$.
	Since $C^{\gamma}$ is a refinement of all coarsenings of $C$ that satisfy $\gamma$, it suffices to show that $C(A) \subseteq \max_{R_C^A}$ 	for all $A$.
	To this end, let $x\in C(A)$ for some $A\in\mathcal{U}$. 
	Then, $x \mathrel{R_C^A} y$ for all $y\in A$.
	In other words, $x\in \max_{R_C^A}(A)$. 
	
	For the inclusion from right to left, let 
	$x\in \max_{R_C^A}(A)$. Then, $x$ is maximal with respect to $R_C^A$ and 
	for all $y\in A$, there is $B_y\in\mathcal{U}$ such that $y\in B_y\subseteq A$ and $x\in C(B_y)$. 
	It follows that $x\in C(B_y)\subseteq C^{\gamma} (B_y)$ since $C\subseteq C^{\gamma}$.
	Moreover, since $C^{\gamma}$ satisfies~$\gamma$, $x\in C^{\gamma} (\bigcup_{y\in A} B_y)=C^{\gamma} (A)$. 
\end{proof}

We consider the $\gamma$-hulls of common social choice functions in \Cref{sec:hull_scf}.

\section{Applications to Social Choice}
\label{sec:socialchoice}

In social choice theory, which studies choice functions that are based on the preference relations of multiple agents, choices between \emph{pairs} of alternatives---or, equivalently, the base relation---are typically given by majority rule. The problem is how these choices should be extended to arbitrary feasible sets in some reasonable and consistent way \citep[see, e.g.,][]{Lasl97a,BBH15a}. 
Few social choice functions are known to satisfy~$\gamma$. The notion of local rationalizability sheds more light on why this is the case. 
Any such function is induced by a family of locally rationalizing relations satisfying the conditions given in \Cref{Def:UpRat}.
Formally, let $\overline{R}$ be a complete relation on $U$ and $C$ be a choice function. We say that $C$ is \emph{based} on $\overline{R}$ if ${\baseRC} = {\baseR}$.

In this section, we will analyze existing social choice functions that satisfy $\gamma$ in terms of their locally rationalizing relations, present axiomatic characterizations of some of these functions, and study procedures for systematically generating social choice functions that satisfy~$\gamma$.
We will view the weak majority relation $\baseR$ as a directed graph and refer to a strict edge $(x,y)$ when $x\base y$. 
As specified in \Cref{Def:UpRat}, the recurring theme is to ``locally'' remove strict edges until all majority cycles are broken.

\subsection{Transitive Local Rationalizability}
\label{sec:tlocal}

In order to circumvent Arrow's impossibility, \citet{Bord76a} has studied a setting where pairwise choices are given and the expansion part of transitive rationalizability ($\beta^+$) extends these choices to feasible sets of cardinality greater than two. \citeauthor{Bord76a} postulates that ``the essence of the rationality concept is that choices on large sets must depend `in a positive way' on choices on two-element sets.''
Eventually, \citeauthor{Bord76a} arrives at a characterization of a social choice function known as the \emph{top cycle}. The top cycle has been reinvented several times and is known under various names such as Good set \citep{Good71a}, Smith set \citep{Smit73a}, weak closure maximality \citep{Sen77a}, and GETCHA \citep{Schw86a}.
We simplify and rewrite Bordes' characterization using transitive local rationalizability. 

Let the base relation $\overline{R}$ be given and define the locally rationalizing relation for a feasible set $A$ as follows.
\[
x\mathrel{R^A} y \quad\text{iff}\quad \text{there are $x_1, \dots , x_k\in A$ s.t. $x= x_1 \baseR \dots \baseR x_k=y$.\tag{Transitive closure}}
\]
This family of relations locally rationalizes the \emph{top cycle}.
Since $R^A$ is defined as the transitive closure of $\overline{R}$, the top cycle naturally satisfies $\beta^+$ and is based on $\baseR$. Moreover, it is the finest such function.

\begin{proposition}
	 Let $\overline{R}$ be a complete relation on $U$. Among all choice functions based on $\baseR$, the top cycle is the finest choice function satisfying transitive local rationalizability.
\end{proposition}
\begin{proof}
	Let $R^A$ be the transitive closure of $\overline{R}$ in $A$ as defined above. By definition, each of these relations is transitive.
	Now let $(\widehat{R}^A)_A$ transitively locally rationalize some choice function $C$ based on $\baseR$. Let $x,y\in A$ such that $x \mathrel{R^A} y$. Then, there is a path $x = x_1 \baseR \dots \baseR x_k=y$ in $A$. Since $\overline{R} \cap (A\times A)\subseteq \widehat{R}^A$, this directly implies $x=x_1 \mathrel{\widehat{R}^A} x_2 \mathrel{\widehat{R}^A} \dots \mathrel{\widehat{R}^A} x_k=y$.
	By transitivity of $\mathrel{\widehat{R}^A}$, we thus have $x\mathrel{\widehat{R}^A} y$. This shows that ${R^A}\subseteq {\mathrel{\widehat{R}^A}}$ and, in particular, implies that 
	the top cycle is a refinement of~$C$.
\end{proof}

\subsection{PIP-Transitive Local Rationalizability}\label{sec:piplocal}
We propose a new majoritarian rule that is locally rationalizable using PIP-transitive relations and a refinement of the top cycle.
Intuitively, we remove an edge $y\base x$ if it is contained in a cycle that may contain indifferences but two indifferences can only appear consecutively right before $y$, i.e., $y \base x \dots x_{k-2}\baseI x_{k-1} \baseI y$ is allowed. Since the indifference relation is reflexive, we, roughly speaking, only delete edges on cycles in which every second edge is strict.

Let the base relation $\overline R$ be given. We define $R^A$ as follows.
\begin{align*}
x\mathrel{R^A} y \quad\text{iff}\quad &\text{there are $x_1, \dots , x_k\in A$ s.t. $x= x_1 \baseR \dots \baseR x_k=y$,}\\
 &\text{where $x_i \baseI x_{i+1}$ implies $x_{i+1} \base x_{i+2}$ for all $i < k-2$}.
\end{align*}

We will refer to the social choice function defined by these relations as the \emph{alternating cycle}.

\begin{proposition}
	The alternating cycle satisfies $\gamma^+$.
\end{proposition}
\begin{proof}
It is easy to check that $(R^A)_A$ satisfies all requirements for local rationalizability and the alternating cycle is thus well-defined.
To show PIP-transitivity, let $a \PA b\RA c\PA d$. We need to show that $a\PA d$, i.e., that no alternating path exists from $d$ to $a$.
By $b\RA c$, there must be some path $P_{b,c}$ from $b$ to $c$ where $\baseI$ can only occur consecutively in the last step.
Now, let $P_{d,x}$ be a path starting from $d$ and ending at any $x\in A$ where $\baseI$ does not occur consecutively, which we denote by a dotted line. 

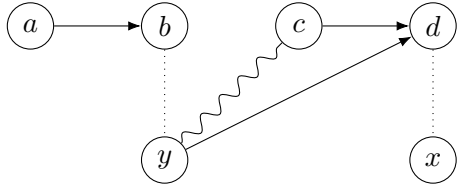
\begin{center}
\begin{tikzpicture}[node distance=\nd,vertex/.style={circle,draw,minimum size=2*\ra,inner sep=1pt}
	]
    \node[vertex] (a) {$\mathwordbox{a}{a}$};
    \node[vertex] (b) [right=of a] {$\mathwordbox{b}{b}$};
    \node[vertex] (c) [right=of b] {$\mathwordbox{c}{c}$};
    \node[vertex] (d) [right=of c] {$\mathwordbox{d}{d}$};
    \node[vertex] (y) [below=of b] {$\mathwordbox{y}{y}$};
    \node[vertex] (x) [below=of d] {$\mathwordbox{x}{x}$};
    
    \draw[-Latex] (a) to (b);
    \draw[-Latex] (c) to (d);
    
    \draw[dotted] (b) to (y);
    \draw[dotted] (d) to (x);
    
    \draw[-,decorate,decoration={snake}] (y) to (c);
    
    \draw[-Latex] (y) to (d);
\end{tikzpicture}
\end{center}

It suffices to show that $a\base x$. We will do this by combining the two existing paths.
First, let $y\in A$ be the last alternative before $c$ on the path $P_{b,c}$. Since $c\PA d$ and $y \baseR c$, it must be that $y\base d$.
Thus, we define the path $P_{b,x}$ by starting at $b$ and walking along $P_{b,c}$ until we arrive at $y$, then using $y\base d$, and finally walk along $P_{d,x}$. Note that since we connect the paths using a strict edge, $P_{b,x}$ has no consecutive occurrences of $\baseI$. If it were that $x\baseI a$, we could extend this path in one step to a path from $b$ to $a$ that has no consecutive occurrences of $\baseI$ except for at the last step. This is not possible by $a\PA b$ and hence $a \base x$, which concludes the proof of PIP-transitivity.
\end{proof}

\subsection{Quasi-Transitive Local Rationalizability}\label{sec:qtlocal}

In this section, we discuss quasi-transitively locally rationalizable social choice functions that fail to be PIP-transitively locally rationalizable: the \emph{Gillies uncovered set}, the \emph{Bordes uncovered set}, the \emph{McKelvey uncovered set}, the \emph{deep uncovered set}, and the \emph{uncaptured set}.

\subsubsection{Gillies Uncovered Set}
Let the base relation $\baseR$ be given and consider the following family of locally rationalizing relations. 
\[
\text{$x\mathrel{R}^A y\quad\text{iff}\quad x \baseR y$ or there is some $z \in A$ with $x\baseR z \base y$.}
\tag{Covering relation}
\]
This relation is known as the Gillies covering relation \citep{Gill59a,Dugg11a} and locally rationalizes a choice function known as the \emph{(Gillies) uncovered set}. For each feasible set, the uncovered set contains the maximal elements with respect to the corresponding covering relation and is thus based on $\baseR$. 
Moreover, we see that for some given $x,y$, we have $x\PA y$ (``$x$ \emph{covers} $y$ in $A$''), if and only if $x \base y$ and for all $z\in A$ with $z \base x$ we have that $z \base y$.
Thus, if some $x$ covers some $y$ in some feasible set, then the same holds true for all feasible subsets containing $x$ and $y$. 
By definition, the covering relation is quasi-transitive.

We now give a characterization of the uncovered set using quasi-transitive local rationalizability and one additional axiom.\footnote{Other characterizations of the uncovered set using inclusion-minimality, either for strict base relations or McKelvey's variant of the uncovered set, were given by \citet{Moul86a}, \citet{DuLa99a}, and \citet{PeSu99a}.}
We say that a choice function $C$ is \emph{weakly idempotent} if 
\[C(C(A))= C(A) \text { for all }A\in \mathcal{U} \text{ with }\lvert C(A) \rvert = 2\text.\tag{Weak idempotency}\]
The stronger---unrestricted---version of idempotency ($C(C(A))= C(A)$ for all $A\in \mathcal{U}$) is not satisfied by the uncovered set (it is, for example, satisfied by the much cruder top cycle; see \Cref{sec:tlocal})

\begin{proposition}
\label{pro:uc}
	Let $\overline{R}$ be a complete relation on $U$. Among all choice functions based on $\baseR$, the uncovered set is the finest one satisfying quasi-transitive local rationalizability and 
	weak idempotency.
\end{proposition}
\begin{proof}
	We first show that the uncovered set satisfies quasi-transitive local rationalizability and 
	weak idempotency.
	By definition, the family of covering relations $(R^A)_A$ locally rationalizes the uncovered set. 
	To show that $R^A$ is quasi-transitive, let $A$ be a feasible set 
	and $x,y,z \in A$ with $x\PA y$ and $y\PA z$. It then needs to be shown that $x \PA z$.
	Since $y$ covers $z$ and $x\base y$, we have that $x \base z$. Further, let $w$ be given with $w\base x$. Since $x$ covers $y$, we also have that $w\base y$. Moreover, since $y$ covers $z$, $w \base z$.  Hence, $x \PA z$.
	Now, assume for contradiction that the uncovered set violates weak idempotency. This implies that $\{x,y\}$ is the uncovered set in some $A\in\mathcal{U}$, even though $x \base y$. Since $x$ cannot cover $y$ in $A$, there has to be some  $z\in A$ such that $y \baseR z \base x$.
	Moreover, since $z$ is not in the uncovered set, it has to be covered by some alternative in $A$.
	Quasi-transitivity of the covering relation implies that $x\PA z$ or $y \PA z$.
	However, both cases are at variance with $z \base x \base y$.
	
	In order to show that the uncovered set is the \emph{finest} choice function satisfying the desired axioms, let $C$ be a choice function based on $\baseR$ that satisfies quasi-transitive local rationalizability and weak idempotency. The family of quasi-transitive locally rationalizing relations is given by $(R^A)_A$.
	It then needs to be shown that the uncovered set is a refinement of $C$.
	To this end, let $x$ be in the uncovered set of $A$. We will show that $x\in C(A)$ by repeatedly applying $\gamma$ on a collection of 2-element and 3-element sets in which $x$ is selected and whose union is $A$. First note that for all $y\in A$ with $x \baseR y$, $x\in C(\{x,y\})$. The more difficult case is that of $y\in A$ with $y \base x$. It follows from the assumption that $x$ is uncovered that there is $z \in A$ such that $z\base y$ and $x \baseR z$.
	Now, let $B\coloneqq \{x,y,z\}$. If $C(B)=\{y\}$, then 
	transitivity of $\PB$ implies $y \PB z$, which contradicts $z \base y$. Similarly, $C(B)=\{z\}$ entails $z \PB x$, which contradicts $x \baseR z$. 
	Moreover, it follows from weak idempotency that $C(B)\neq \{y,z\}$. 
	As a consequence, $x\in C(B)$.
	We have thus found a $B_y\in\mathcal{U}$ for every alternative $y\in A$ such that $x,y\in B_y$ and $x\in C(B_y)$. Repeated application of $\gamma$ then shows that $x\in C(A)$. 
\end{proof}

\subsubsection{Further Variants of the Uncovered Set}
Several variants of covering relations and uncovered sets have been proposed in the literature \citep[see, e.g.,][]{Bord83a,BrFi08b,Dugg11a}. We follow Duggan's terminology here.
The \emph{Bordes uncovered set} is defined by replacing the second condition in the definition of the covering relation with `$x\base z \baseR y$'. Similar to the proof of \Cref{pro:uc}, it can be shown that the Bordes uncovered set is the finest choice function satisfying quasi-transitive local rationalizability and a technical condition requiring that for all $A\in \mathcal{U}$ with $|A|=3$,
\[
\text{if $y\in C(A)$ and $x\baseC y$,\quad then $x\in C(A)$.}
\tag{$\ast$}
\]
The \emph{McKelvey uncovered set} is defined by replacing the second condition in the definition of the covering relation with `$x\baseR z \base y$ or $x\base z \baseR y$'.
It follows from a characterization by \citet{PeSu99a} that the McKelvey uncovered set is the finest choice function satisfying quasi-transitive local rationalizability, weak idempotency, and~($\ast$).
The Mc\-Kelvey uncovered set contains both the Gillies and the Bordes uncovered set.
\citet{Dugg11a} introduced a coarsening of the McKelvey uncovered set that he calls the \emph{deep uncovered set}. Here, the second condition in the definition of the covering relation is replaced with `$x\baseR z \baseR y$'. \Cref{pro:uc} can be adapted to characterize the deep uncovered set: it is the finest choice function satisfying quasi-transitive local rationalizability and a strengthening of~($\ast$) which requires that for all $A\in \mathcal{U}$ with $|A|=3$, $y\in C(A)$ and $x\baseRC y$ implies $x\in C(A)$.\footnote{
Other variants of the uncovered set that take into account weights of pairwise comparisons have been studied by \citet{DuLa99a} and \citet{PeDe17a}.
}

\subsubsection{Uncaptured Set}

The following rule proposed by \citet{Dugg06a} aims at 
{removing edges from three-cycles in which the preceding or subsequent edge is strict, and in four-cycles in which both the preceding and subsequent edge are strict.}
Formally, define the family of local rationalizing relations for a feasible set $A$ as follows. 
\begin{align*}
x\mathrel{R}^A y\quad\text{iff} \quad x \baseR y \text{ or there are $z,w \in A$ with } &x\baseR z \base y \text{, or}\\
 &x\base z \baseR y \text{, or}\\
 &x\base z\baseR w\base y.
\end{align*}
Duggan refers to the resulting social choice function as the \emph{uncaptured set}. It satisfies both $\gamma$ and $\epsilon^+$, which can be demonstrated by proving that for all feasible sets $A$, $R^A$ is quasi-transitive.
Let $A$ be a feasible set with $x\PA y \PA z$ for some $x,y,z \in A$. We need to show $x \PA z$, which consists of three parts.
First, we must exclude three-cycles with a strict edge ending in $x$. Let $a\in A$ with $a\base x$. Then, by $x\PA y$, $a\base y$.
Now, $y\PA z$ implies $a \base z$. Second, we exclude three-cycles with a strict edge starting in $z$. Let $a\in A$ with $z\base a$.
Then, $y \base a$ by $y\PA z$. $x \PA y$ now implies $x\base z$. Last, we need to exclude four-cycles in which only the opposing edge can be weak. Let $a,b\in A$ such that $z\base a \baseR b$. Since $y\PA z$, it must be that $y\baseR b$. Thus, $x\baseR b$ by $x\PA y$, concluding the proof.

\subsection{Local Rationalitability}

In this section, we discuss three social choice functions that are locally rationalizable but fail to be quasi-transitively locally rationalizable: the \emph{untrapped set}, \emph{two-stage majoritarian choice}, and \emph{split cycle}.

\subsubsection{Untrapped set}
Let the base relation $\overline{R}$ be given and define the local rationalizing relation for a feasible set $A$ as follows.
\[
x\mathrel{R^A} y \quad\text{iff}\quad \text{$x\baseR y$ or there are $x_1, \dots , x_k\in A$ s.t. $x= x_1 \base \dots \base x_k=y$.}
\]
This is known as the trapping relation \citep{Dugg06a}  and induces the \emph{untrapped set}.
To prove well-definedness, first let feasible sets $B\subseteq A$ with $x,y\in B$ and $x \RB y$ be given. By definition, it must be that $x \RA y$.
For acyclicity, let $A$ be a feasible set with $x_1,\dots, x_k\in A$ and $x_1\PA \dots \PA x_k$. It needs to be shown that $x_1\RA x_k$.
Observe that $x \PA y$ if and only if $x \base y$ and there is no path w.r.t. $\base$ in $A$ from $y$ to $x$. Consequently, it must be that ${\PA} \subseteq {\base}$, which directly implies $x_1 \base \dots\base x_k$. As desired, this leads to $x_1\RA x_k$ by definition of $R^A$.\footnote{
The untrapped set is conceptually an expansion-consistent version of GOTCHA aka the Schwartz set \citep{Schw86a}. The latter chooses the maximal elements with respect to the transitive closure of the strict majority relation. However, the closure does not respect the weak base relation, and the rule hence violates $\gamma$.
}

\subsubsection{Two-stage majoritarian choice}
 
Let $\overline{R}$ be majority rule for a given preference profile and ${\ge}\subseteq {U\times U}$ a fixed linear order over $U$. Now define the local rationalizing relation for feasible set $A$ by letting
\[
x\mathrel{R^A} y \quad\text{iff}\quad x \baseR y \text{ or ($x > y$ and there is $z\in A$ s.t.~$z\base y$ and $z>y$).}
\]
Clearly, for two feasible sets $A$ and $B$ with $B\subseteq A$, $x \RB y$ implies $x \RA y$.
For the strict part of the local preference relation, we have $x \PA y$ if and only if $x \base y$ and ($x>y$ or $x\in \max_{{\base}\cap{>}}(A)$).
Now assume for contradiction that $R^A$ is not acyclic for some $A$. Then we have $x_1\PA \dots \PA x_k \PA x_1=x_{k+1}$. For $B\coloneqq \{x_1, \dots, x_k\}$, we still have $x_1\PB \dots \PB x_k \PB x_1$ by Condition \emph{(iii')}. Now consider two cases. First, assume there is some $i$ with $x_i > x_{i+1}$. By renaming the alternatives we can assume without loss of generality that $i=1$. Then, $x_2\not\in \max_{{\base}\cap{>}}(B)$ and hence $x_2>x_3$. It follows by induction that $x_2 > x_3> \dots> x_k > x_1$, which contradicts the acyclicity of $\ge$. Otherwise, $x_i\in \max_{{\base}\cap{>}}(B)$ for all $i$ with $1\le i \le k$. Strictness of $\ge$ then implies that $x_1<x_2<\dots <x_k< x_1$, which again contradicts the acyclicity of~$\ge$. 
The choice function locally rationalized by $(R^A)_A$ is based on $\overline{R}$. It was recently introduced for antisymmetric $\baseR$ by \citet{HoSp21a} as \emph{two-stage majoritarian choice}. We lift the restriction on the base relation while maintaining local rationalizability. Note that for antisymmetric relations, the choice function is single-valued for all feasible sets.

\subsubsection{Split cycle}

In order to define split cycle, we not only need the base relation $\overline{R}$ but also an order over pairs of alternatives ${\trianglerighteq}\subseteq {(U\times U)^2}$. These orders occur naturally in social choice and are typically induced by majority margins, i.e., the number of agents who prefer the former to the latter minus the number of agents who prefer the latter to the former.

Then we can define a locally rationalizing relation for feasible set $A$ by letting\footnote{
Observe that if $x\base y$, then $(x,y)\triangleright(v,w)$ for all $v,w$ with $v\baseI w$.
Hence, in contrast to the $\gamma$-core, an equivalent definition is to demand that $x= x_1 \base \dots \base x_k=y$.
}
\begin{align*}
x\mathrel{R}^A y\quad\text{iff}\quad & x \baseR y \text{ or there are $x_1, \dots , x_k\in A$ s.t. } x= x_1 \baseR \dots \baseR x_k=y \text{ and }\\
&(x_i,x_{i+1})\trianglerighteq (y,x) \text{ for all }i<k \text.
\end{align*}
By contraposition, this means $x \PA y$ if and only if $x \base y$ and for all weak base relation cycles $(x_1,\dots, x_k)$ in $A$ with $x_1=y$, $x_k=x$, there is some $i<k$ with $(x,y) \triangleright (x_i,x_{i+1})$.
Clearly, if $x\PA y$, then $x \PB y$ for all $B\subseteq A$ with $x,y\in B$. Further, $P^A$ is acyclic. To see this, let $x_1\PA \dots \PA x_k$. Assume for contradiction $x_k \PA x_1$. Then we have a base relation cycle $(x_1,\dots x_k)$ in $A$. Setting $x_{k+1}=x_1$, one of the pairs $(x_i, x_{i+1})$ must be minimal with respect to $\trianglerighteq$, a contradiction to $x_i\PA x_{i+1}$. 
The choice function locally rationalized by $(R^A)_A$, based on $\baseR$, is called \emph{split cycle} and was recently proposed by \citet{HoPa20a}. Split cycle is similar to Tideman's ranked pairs \citep{Tide87a} and Schulze's rule \citep{Schu11a}, but satisfies some additional desirable properties, most notably~$\gamma$. \Cref{thm:Gamma} provides an intuitive explanation of this behavior.

In the following, we give a characterization of split cycle using local rationalizability. It is inspired by a characterization of split cycle as a ``variable-election collective choice rule (VCCR)'' due to \citet{HoPa21a}. 
The theory of local rationalizability allows us to split their axiom ``coherent IIA'' into two parts, namely local rationalizability and ``crucial defeat''. The advantage of this representation is that local rationalizability is the only required variable-agenda axiom while crucial defeat can be formulated for a fixed agenda. By \Cref{thm:Gamma}, the former condition is equivalent to $\gamma$, a well-established axiom in social choice theory.

To prepare the characterization, let $\mathcal{O}$ denote the set of all complete, transitive orders $\trianglerighteq$ on $U\times U$, 
such that $(x,y) \trianglerighteq (v,w)$ if and only if $(w, v) \trianglerighteq (y, x)$ for all $x,y,v,w\in U$.
Then each ${\trianglerighteq} \in \mathcal O$ induces a complete relation on $U$ by setting $x \mathrel{\overline R_\trianglerighteq} y$ if and only if $(x,y)\trianglerighteq (y,x)$. We then say that a choice function $C$ is based on $\trianglerighteq$, if ${\baseRC} \mathrel{=} \overline{R}_\trianglerighteq$. The strict part of $\trianglerighteq$ is denoted by $\triangleright$, the symmetric part by $\doteq$.
As mentioned above, such transitive orders naturally arise via majority margins in the context of social choice. A majority tie between $x$ and $y$ then corresponds to $(x,y)\doteq (y,x)$. If, instead, $x$ would beat $y$ with a majority margin larger than the one with which $v$ beats $w$, then this corresponds to $(x,y)\triangleright (v,w) \triangleright (w,v)\triangleright(y,x)$.

An $\emph{extended choice function}$ is a family of choice functions $(C(\,\cdot,\trianglerighteq))_{\trianglerighteq \in \mathcal{O}}$, such that for each $\trianglerighteq$, $C(\,\cdot,\trianglerighteq)$ is based on $\trianglerighteq$. 
We say that an extended choice function $(C(\,\cdot,\trianglerighteq))_{\trianglerighteq \in \mathcal{O}}$ satisfies local rationalizability if each $C(\,\cdot,\trianglerighteq)$ is locally rationalizable. Split cycle is denoted by $\splitcycle(\,\cdot,\trianglerighteq))_{\trianglerighteq \in \mathcal{O}}$.
The characterization of split cycle requires two axioms on top of local rationalizability. The first one is the familiar notion of \emph{neutrality}: for all permutations $\pi$ on $U$, $C(\pi(A), \pi(\trianglerighteq))= \pi (C(A,\trianglerighteq))$.\footnote{By $\pi (\trianglerighteq)$ we denote the relation $\trianglerighteq'$ such that $(x,y)\trianglerighteq (v,w)$ if and only if $(\pi(x),\pi(y))\trianglerighteq' (\pi(v),\pi(w))$.}

In order to define the second axiom, we need some more terminology.
Let $\trianglerighteq$, ${\trianglerighteq^*} \in \mathcal O$ and $x,y\in U $ be given, such that $x \mathrel{\overline R_\trianglerighteq} y$ and $x \mathrel{\overline R_{\trianglerighteq^*}} y$.
Further, let $\trianglerighteq$ be identical to $\trianglerighteq^*$ in all comparisons where $(x,y)$ and $(y,x)$ are not involved. 
We then say that $\trianglerighteq^*$ is obtained from $\trianglerighteq$ by \emph{decreasing the intensity} of $(x,y)$, if for all $e \in U\times U$, we have that 
\[
e \trianglerighteq (x,y)\text{ implies }e \trianglerighteq^* (x,y) \quad\text{ and }\quad
e \triangleright (x,y)\text{ implies }e \triangleright^* (x,y)\text.
\]
The axiom of \emph{crucial defeat} then requires that for all $A\in\mathcal{U}$ and ${\trianglerighteq} \in \mathcal O$, if $x\notin C(A, \trianglerighteq)$, then there is some $y\in A$ with $(y,x) \triangleright (x,y)$ such that for any $\trianglerighteq^*$ obtained by decreasing the intensity of some $e\neq (y,x)$, $x\notin  C(A,\trianglerighteq^*)$.

\begin{proposition}
	\thlabel{Thm:SCCharacterization}
	Split cycle is the finest extended choice function satisfying local rationalizability, crucial defeat, and neutrality.
\end{proposition}
\begin{proof}
	Let $S^A$ denote the strict part of the split cycle relation on $A\in\mathcal{U}$ w.r.t. $\overline R_\trianglerighteq$. This means that $x \mathrel{S^A} y$ iff in any cycle in $A$ containing $(y,x)$, there is some $e\in A\times A$ such that $(y,x) \triangleright e$ and $e$ is contained in the cycle.
	First, we show that split cycle satisfies the axioms. Local rationalizability follows immediately from the definition.
	For neutrality, let $\pi: U\to U$ be a permutation, ${\trianglerighteq}\in \mathcal O$ and $A$ a feasible set. Let $x'\in A$. Then $x'\in \pi(\splitcycle(A,\trianglerighteq))$, iff there is $x\in A$ with $\pi(x)=x'$ and $x\in \splitcycle(A,\trianglerighteq)$. By definition, it is equivalent that for all $y\in A$ with $y \mathrel{P_\trianglerighteq} x$, there must be some cycle $C\subseteq A$ with respect to $R_\trianglerighteq$ containing $(y,x)$, such that for all $e\neq (y,x)$ contained in the cycle we have $e\trianglerighteq (y, x)$. Let $y'\coloneqq \pi (y)$. Then, this is equivalent to the existence of a cycle $C'$ with respect to $R_{\pi(\trianglerighteq)}$ such that for all $e'\neq (y',x')$ we have $e'\mathrel{\pi(\trianglerighteq)}(y',x')$. This holds with the identities $C' = \pi(C)$ and $e' = \pi(e)$. 
	By definition, this is equivalent to $x'\in \splitcycle(\pi(A), \pi(\trianglerighteq))$.
	As desired, this implies $\splitcycle(\pi(A), \pi(\trianglerighteq))= \pi (\splitcycle(A,\trianglerighteq))$.
	For crucial defeat, let $x\in A$, $x\notin \splitcycle(A, \trianglerighteq)$. 
	Then, there must be $y\in A$ with $y\mathrel{S^A} x$ (w.r.t. $\trianglerighteq$). Let $\trianglerighteq^*$ be obtained from $\trianglerighteq$ by decreasing the intensity of some $e^*\neq (y,x)$. It suffices to show that $y \mathrel{S^A} x$ (w.r.t. $\trianglerighteq^*$). To this end, consider an arbitrary cycle in $A$ containing $(y,x)$ and let $e\in A\times A$ such that $(y,x) \triangleright e$ and $e$ is contained in the cycle. If  $e^*\neq e$, then it must be that $(y,x) \triangleright^* e$, since the relations are identical up to comparisons involving $e^*$. Otherwise, $e^* = e$. Then $(y,x) \triangleright e^*$ implies $(y,x) \triangleright^* e^*$. Since the cycle containing $(y,x)$ in $A$ was arbitrary and by definition of split cycle, $x\notin \splitcycle(A, \trianglerighteq^*)$.
	
	Now, let $C$ be an extended choice function satisfying the axioms. It suffices to show that split cycle is a refinement of $C$.
	Assume for contradiction that there are ${\trianglerighteq}\in \mathcal O$, $A\in\mathcal{U}$, and $x\in A$ such that $x\in \splitcycle(A,\trianglerighteq)$ but $x\notin C(A,\trianglerighteq)$. Let $y\in A$ be given by crucial defeat.
	By assumption, it is impossible that $y\mathrel{S^A} x$. 
	Further, first assume without loss of generality that there is no $z\neq y$ in $A$ with $z\mathrel{\overline{P}_{\trianglerighteq}} x$.
	Then there must be some cycle $(y,x,x_2,\dots, x_k, y)$ in $A$ in which $(y,x)$ is minimal w.r.t. $\trianglerighteq$. Now, we modify $\trianglerighteq$ to arrive at some $\trianglerighteq^*$ as follows. We decrease the intensity of each $e=(x_i,x_{i+1})$ such that $e\doteq^* (y,x)$. For all other $e=(v,w)\in U\times U$ with $e\neq (y,x), (x,y)$, we decrease their intensity until $(v,w)\doteq^* (w,v)$ to finally obtain $\trianglerighteq^*$. Note that this is possible, as after each decrement, by assumption, only $y$ is a feasible candidate for crucial defeat, and we are therefore allowed to keep decreasing the intensity of any $e\neq (y,x), (x,y)$. By crucial defeat, it must be that $x\notin C(A,\trianglerighteq^*)$. 
	By construction, we have $ x\mathrel{\overline{R}_{\trianglerighteq^*} z}$ for all $z \in A\setminus\{y\}$.
	By local rationalizability, $y \PAC x$. Let $B\coloneqq \{y,x,x_2,\dots, x_k\}\subseteq A$.
	Then, $x\notin C(B,\trianglerighteq^*)$. By neutrality, $z\notin C(B,\trianglerighteq^*)$ for all $z\in B$. This is a desired contradiction to the nonemptiness of $C$.
	
	In case that there is $z_1\neq y_1$ with $y_1\mathrel{\overline{P}_{\trianglerighteq}} x$ and $z_1\mathrel{\overline{P}_{\trianglerighteq}} x$, we apply crucial defeat.
	This allows us to decrease the intensity to $(z_1,x)\doteq^1 (x,z_1)$ while still having $x\notin C(A, \trianglerighteq^1)$.
	Once more by crucial defeat, there must be $y_2 \mathrel{\overline{P}_{\trianglerighteq ^1}} x$. If there still is $z_2\neq y_2$, we can again decrease the intensity to $(z_2,x)\doteq^2 (x,z_2)$ while maintaining $x\notin C(A, \trianglerighteq^2)$.
When iterating this argument, finiteness of the feasible set implies that we eventually terminate with only one element left.
	   
\end{proof}

\subsection{Gamma Hulls of Social Choice Functions}
\label{sec:hull_scf}

The~$\gamma$-hull, as defined in \Cref{def:gammahull} and characterized in \Cref{thm:gammahull}, allows us to turn any social choice function that violates~$\gamma$ into its finest coarsening that satisfies~$\gamma$. Perhaps surprisingly, the~$\gamma$-hulls of many common social choice functions are different from their known coarsenings that satisfy~$\gamma$ and thus yield new unexplored social choice functions. This is, for example, the case for the~$\gamma$-hulls of Borda's rule, Copeland's rule, and the essential set, all of which are different from each other and any rule known to us \citep[see, e.g.,][for definitions]{BBH15a,FHN15a}. It follows from the well-known inclusion of Copeland's rule in the top cycle that the~$\gamma$-hull of Copeland's rule is contained in the top cycle (see \Cref{sec:tlocal}). Furthermore, it can be shown that the~$\gamma$-hull of the essential set is sandwiched in between the Gillies uncovered set and the McKelvey uncovered set (see \Cref{sec:qtlocal}).
On the other hand, the~$\gamma$-hull of the omninomination rule (which returns all top-ranked alternatives) is the Pareto rule. This example illustrates two observations. First, every social choice function that satisfies $\alpha$ has a rationalizable~$\gamma$-hull. Secondly, if a social choice function satisfies $\epsilon^+$, then so does its~$\gamma$-hull.

\subsection{Gamma Core of Social Choice Functions}

The $\gamma$-hulls of common social choice functions are rather coarse.
In this section, we propose another method to define choice functions that satisfy $\gamma$.
	
	To this end, let $\baseR$ be given and consider an arbitrary choice function $f$. The idea of $f$ is to identify the \emph{worst} elements in a feasible set and then break each cycle by removing the outgoing edges of alternatives returned by $f$.
	Formally, we define a locally rationalizing relation for feasible set $A$ by letting
	\begin{align*}
	x\mathrel{R^A} y \quad\text{iff}\quad x \baseR y &\text{ or there are $x_1, \dots , x_k\in A$ s.t. $x= x_1 \base \dots \base x_k=y$ and} \\
						&y\in f( \{x_1,\dots,x_k \} )\text{.}
	\end{align*}
	$(R^A)_A$ locally rationalizes 
	the choice function $\dot{C}_f(A)\coloneqq \max_{R^A} A$, which we refer to as the \emph{strict~$\gamma$-core} of~$f$.
	By definition, $\dot{C}_f$ is always a refinement of the untrapped set (see \Cref{sec:tlocal}). As a consequence, $\dot{C}_f$ fails to be quasi-transitive and thus violates $\epsilon^+$.
	This can be circumvented by modifying the definition of the locally rationalizing relations to
	\begin{align*}
	x\mathrel{R^A} y \quad\text{iff}\quad x \baseR y &\text{ or there are $x_1, \dots , x_k\in A$ s.t. $x= x_1 \baseR \dots \baseR x_k=y$ and} \\
						&y\in f( \{x_1,\dots,x_k \} )\text{.}
	\end{align*}
	The resulting choice function $C_f(A)\coloneqq \max_{R^A} A$ is called the \emph{(weak)~$\gamma$-core} of $f$.
	We then have $\dot{C}_f\subseteq C_f$, where the latter is always a refinement of the top cycle. Many existing social choice functions can be written as the $\gamma$-core of natural choice functions~$f$:

	\begin{itemize}
		\item If $f(A)=A$ for all $A\in\mathcal{U}$, then $C_f$ is the top cycle.
		\item	If $\displaystyle f(A)= \arg\min_{x\in A} \lvert \{y\in A\colon x\baseR y \}\vert$, then $C_f$ is the Gillies uncovered set. 
		\item 	If $f(A)=\max_{\geq} A$ for some fixed linear order ${\geq}\subseteq U\times U$, then $C_f$ is two-stage majoritarian choice.
	\end{itemize}

	The $\gamma$-core procedure also enables the definition of new social choice functions satisfying $\gamma$. For example, already in small instances we observed that the $\gamma$-core of the choice function that returns all alternatives with minimal Borda score is different from any other social choice function that we are familiar with.

\section*{Acknowledgments}
This material is based on work supported by the Deutsche Forschungsgemeinschaft under grants {BR~2312/11-2} and {BR~2312/12-1}. Preliminary results from this paper were presented at the 7th International Conference on Algorithmic Decision Theory (November 2021) and the 16th Meeting of the Society of Social Choice and Welfare in Mexico City (June 2022). The authors thank Martin Bullinger and Sean Horan for helpful feedback.

\end{document}